\documentclass{article}

\usepackage{amsthm}  %% to get \theroemstyle
\usepackage{amsmath, amssymb, mathtools}
\usepackage{color}
\usepackage{colortbl}
\usepackage{fullpage}
\usepackage{algorithm}
\usepackage[noend]{algpseudocode}
\usepackage{thmtools, thm-restate}
\usepackage{tikz}
\usepackage{enumerate,enumitem}
\usepackage{thm-restate}
\usepackage{hyperref}
\usepackage{graphicx}
\usepackage{amsopn}
\usepackage{caption}
\usepackage{hyperref}
\usepackage{subcaption}
\usepackage{enumitem,linegoal}
\usepackage{comment}
\usepackage{float}
\usepackage[super]{nth}
\usepackage{subcaption}

\usepackage{natbib}
\usepackage{mdframed}
\usepackage{thm-restate}

\newtheorem{theorem}{Theorem}[section]
\newtheorem{lemma}[theorem]{Lemma}

 \newtheorem{proposition}[theorem]{Proposition}
 \newtheorem{corollary}[theorem]{Corollary}

\theoremstyle{definition}
 \newtheorem{definition}[theorem]{Definition}
\newcommand{\R}{\mathbb{R}}
\newcommand{\N}{\mathcal{N}}
\newcommand{\M}{\mathcal{M}}
\newcommand{\I}{\mathcal{I}}
\newcommand{\V}{\mathcal{V}}

\newcommand{\ef}{\mathsf{EF}}
\newcommand{\efx}{\mathsf{EFX}}

\newcommand{\mms}{\mathsf{MMS}}

\newcommand{\eefx}{\mathsf{EEFX}}

\newcommand{\HA}[1]{\textcolor{magenta}{(Hana: #1)}}

\definecolor{mygreen}{RGB}{20,140,80}
\definecolor{mylightgray}{RGB}{230,230,230}

\definecolor{mygreen}{RGB}{20,140,80}
\definecolor{mydarkgray}{gray}{0.15} 
\definecolor{oceanblue}{HTML}{2c55c2}
\hypersetup{
     colorlinks=true,
     citecolor= mygreen,
     linkcolor= black
}

\usepackage[preprint]{neurips_2024}

\usepackage[utf8]{inputenc} % allow utf-8 input
\usepackage[T1]{fontenc}    % use 8-bit T1 fonts
\usepackage{hyperref}       % hyperlinks
\usepackage{url}            % simple URL typesetting
\usepackage{booktabs}       % professional-quality tables
\usepackage{amsfonts}       % blackboard math symbols
\usepackage{nicefrac}       % compact symbols for 1/2, etc.
\usepackage{microtype}      % microtypography
\usepackage{xcolor}         % colors

\title{\boldmath Achieving Maximin Share and $\efx/\ef1$ Guarantees Simultaneously}

\author{%
    Hannaneh Akrami\\
   Max Planck Institute for Informatics\\
   Graduiertenschule Informatik, 
   Universit\"at des Saarlandes\\
   \texttt{hakrami@mpi-inf.mpg.de} \\
   \And
   Nidhi Rathi\\
   Max Planck Institute for Informatics\\
   Saarland Informatics Campus\\
   \texttt{nrathi@mpi-inf.mpg.de}
}

\begin{document}

\maketitle
\begin{abstract}
    We study the problem of computing \emph{fair} divisions of a set of indivisible goods among agents with \emph{additive} valuations. For the past many decades, the literature has explored various notions of fairness, that can be primarily seen as either having \emph{envy-based} or \emph{share-based} lens. For the discrete setting of resource-allocation problems, \emph{envy-free up to any good} ($\efx$) and \emph{maximin share} ($\mms$) are widely considered as the flag-bearers of fairness notions in the above two categories, thereby capturing different aspects of fairness herein. Due to lack of existence results of these notions and the fact that a good approximation of $\efx$ or $\mms$ does not imply particularly strong guarantees of the other, it becomes important to understand the compatibility of $\efx$ and $\mms$ allocations with one another. 
%(which has been the basis of almost all of the previous works studying approximate $\mms$ allocations)

    In this work, we identify a novel way to simultaneously achieve $\mms$ guarantees with $\efx$/$\ef1$ notions of fairness, while beating the best known approximation factors \citep{chaudhury2021little,amanatidis2020multiple}. Our main contribution is to constructively prove the existence of (i) a partial allocation that is both $2/3$-$\mms$ and $\efx$, and (ii) a complete allocation that is both $2/3$-$\mms$ and $\ef1$. Our algorithms run in pseudo-polynomial time if the approximation factor for $\mms$ is relaxed to $2/3-\varepsilon$ for any constant $\varepsilon > 0$ and in polynomial time if, in addition, the $\efx$ (or $\ef1$) guarantee is relaxed to $(1-\delta)$-$\efx$ (or $(1-\delta)$-$\ef1$) for any constant $\delta>0$. In particular, we improve from the best approximation factor known prior to our work, which computes partial allocations that are $1/2$-$\mms$ and $\efx$ in pseudo-polynomial time \citep{chaudhury2021little}.
    %Thereby, we improve the current best approximations known for simultaneously achieving both $\mms$ and $\efx$/$\ef1$ guarantees 
\end{abstract}
% with at most $n-1$ goods given to charity such that no agent envies the charity, 

\section{Introduction}
%\todo{define approximate $\ef1$/$\efx$} \NR{Maybe not, multiple birds with one stone doesn't.}
The theory of fair division addresses the fundamental problem of allocating a set of resources among a group of individuals, with varied preferences, in a meaningfully \emph{fair} manner. The need of fairness is a key concern in the design of many social institutions, and it arises naturally in multiple real-world settings such as division of inheritance% \cite{PrattZ90}
, dissolution of business partnerships, divorce settlements% \cite{BramsT96}
, assigning computational resources in a cloud computing environment, course assignments% \cite{BudishC10}
, allocation of radio and television spectrum, air traffic management, to name a few (see \cite{PrattZ90,brams1996fair,BramsT96,moulin2004fair,BudishC10}). 

The numerous applications depicting the necessity of understanding the process of fairly dividing resources among multiple economic players has given rise to a formal theory of \emph{fair division}. Such problems lie at the interface of economics, mathematics, and computer science, and they have been extensively studied for past several decades \citep{moulin2019fair}. Although the roots of fair division can be found in antiquity, for instance, in ancient Greek mythology and the Bible, its formal history stretches back to the seminal work of Steinhaus, Banach and Knaster in 1948 \citep{steinhaus1948problem}. 

In this work, we study the well-studied fair division setting of allocating a set of discrete or indivisible items among agents. A fair division instance consists of a set  $\N=[n]$ of $n$ agents and a set $\M$ of $m$ items. Every agent $i$ specifies her preferences over the items via an \emph{additive} valuation function $v_i:2^\M \rightarrow \R_{\geq 0}$. The goal is to find a partition $X = (X_1, \dots, X_n)$ of the given items such that every agent $i \in \N$, upon receiving bundle $X_i$, considers $X$ to be \emph{fair}.

Primarily, there have been two ways of defining fairness for resource-allocation settings: (i) \emph{envy-based}, where an agent \emph{compares} her bundle with other bundles in the allocation to decide if it is \emph{fair} to her and (ii) \emph{share-based}, where an agent considers an allocation to be \emph{fair} for her through the value she obtains from her bundle (irrespective of what others receive). \emph{Envy-freeness} \citep{foley1966resource} is arguably the flag-bearer of envy-based notions of fairness that entails an allocation $X$ to be \emph{fair} if every agent values her own bundle at least as much as she value any other agent's bundle (i.e., $v_i(X_i) \geq v_i(X_j)$ for all $i,j \in [n]$). On the other hand, \emph{proportionality} is an important share-based notion of fairness that entails an allocation to be fair when every agent $i \in [n]$ values her bundle at least as much as her proportional share value of $v_i(\M)/n$. Both of these notions are known to exist in the setting where the resource is divisible (i.e., a cake $[0,1]$), but unfortunately, a simple instance where a single valuable (indivisible) item is to be divided between two agents does not admit any envy-free or proportional allocation. 

Within the last decade, we have seen an extensive study of various relaxations of envy-freeness and proportionality that are more suitable for the discrete setting. Among those, the most prominent relaxations, and the focus of our work, include the notions of \emph{envy-freeness up to any good} ($\efx$) \citep{caragiannis2016unreasonable}, \emph{envy-freeness up to one good} ($\ef1$) \citep{lipton2004approximately}, and \emph{maximin share} fairness ($\mms$) \citep{budish2011combinatorial}. Here, $\efx$ and $\ef1$ relax the notion of envy-freeness, while maximin share is considered to be a natural relaxation of proportionality.

\emph{Maximin share} ($\mms$) has been one of the most celebrated relaxations of proportionality for the discrete setting. The maximin share value ($\mms^n_i(\M)$ or $\mms_i$) of an agent $i$, is defined to be the maximum value she can obtain among all possible allocations of the set of items $\M$ among $n$ agents, while receiving the minimum-valued bundle in any allocation.
%we let agent $i$ partition the goods into $n$ bundles wherein she receives her minimum-valued bundle. Going through all possible $n$-partitions, the maximum value she can obtain is her maximin share. We consider an allocation $X$ to be $\mms$ if every agent $i \in [n]$ receives a bundle $X_i$ such that $v_i(X_i) \geq \mms_i$. 
Since $\mms$ allocations may not always exist for fair division instances with more than two agents \citep{procaccia2015cake}, a significant amount of research has been focused on achieving better approximation guarantees (i.e., for some $\alpha \in (0,1)$, every agent $i$ gets a value of $\alpha$-$\mms_i$) for maximin share. A recent breakthrough proves the existence and develops a PTAS to compute $(\frac{3}{4}+\frac{3}{3836})$-$\mms$ allocations for additive valuations \citep{akrami2024breaking}.

On the other hand, focusing on relaxations of envy-freeness, in an $\ef1$ allocation $X$, any agent $i \in [n]$ may envy another agent $j$, but the \emph{envy} must vanish after removing some good from the bundle $X_j$ (i.e., $v_i(X_i) \geq v_i(X_j \setminus \{g\})$ for some $g \in X_j$). $\ef1$ allocations are known to exist and can be computed in polynomial time as well \citep{lipton2004approximately}. Later, \citet{caragiannis2016unreasonable} introduced a stronger relaxation of envy-freeness called $\efx$. Here, again, in an $\efx$ allocation, any agent $i \in [n]$ may envy another agent $j$, but the envy now must vanish after removing \emph{any} good from the bundle $X_j$ (i.e., $v_i(X_i) \geq v_i(X_j \setminus \{g\})$ for all $g \in X_j$). As a complete contrast to $\ef1$, the notion of $\efx$ is fundamentally more challenging and despite significant efforts, the community has not been able to fully understand the existential and computational guarantees of $\efx$ allocations. For instance, the biggest open problem in fair division is to resolve the existence of $\efx$ allocations for instances with four or more agents \citep{procaccia2020technical}. 

Reasonably enough, several approximations and relaxations of $\efx$ have been extensively studied (see Section \ref{sec:related} for more details). One of the notable results herein is pseudo-polynomial time computability of \emph{partial} $\efx$ allocations where at most $n-1$ goods go to \emph{charity} (i.e., remain unallocated) such that no agent envies the charity bundle \citep{chaudhury2021little}.% Inspired by the work of \citet{ABCGL18}, a recent work by \citet{Caragiannis2023} has defined an interesting and useful relaxation of $\efx$, called \emph{epistemic $\efx$} ($\eefx$). $\eefx$ allocations are guaranteed to exist for instances with monotone valuations \cite{AR24eefx} and can be computed in polynomial time for instances with additive valuations \citep{Caragiannis2023}.

%This work focuses on understanding the compatibility of two different classes of fairness notions for indivisible goods - $\ef1$ and $\efx$ with $\mms$. It is important to note that envy-freeness implies proportionality in the setting of a divisible resource. This property no longer holds true for their analogous relaxations in the discrete setting. In particular, neither of $\ef1$ nor $\efx$ allocations imply $\mms$ guarantees, and in fact, it can be as bad as $1/2$-$\mms$. Example. 

It is relevant to note that the notions of $\ef1$/$\efx$ and $\mms$ capture different aspects of fairness. Either of $\ef1$/$\efx$ or $\mms$ properties does not necessarily imply particularly strong approximation guarantees for the other(s) \citep{comparing-efx}. In Section \ref{sec:relation}, we discuss the guarantees $\efx$/$\ef1$ allocations can provide for $\mms$ and vice versa. This is in complete contrast to the divisible setting guarantees, where any envy-free allocation is necessarily proportional as well. Hence, it becomes compelling to ask for allocations that attain good guarantees with respect to envy-based and share-based notions of fairness simultaneously. There are few works along these lines in the literature, e.g., \cite{barman2018groupwise,McGlaughlinG20}, some of which give purely existential guarantees \citep{caragiannis2016unreasonable}. Motivated by the above question, this work focuses on understanding the compatibility of two different classes of fairness notions, i.e., in particular, $\mms$ with $\ef1/\efx$. %And, as a corollary, we also get simultaneous guarantees on approximate $\eefx$ and $\ef1$ as well that partially addresses an open problem stated in \citet{AR24eefx}. %The closest work to ours is by \cite{amanatidis2020multiple} where the authors develop an efficient algorithm to compute an allocation that is simultaneously $0.553$-$\mms$ and $0.618$-$\efx$ for instances with additive valuations. \HA{Maybe we want to focus more on the charity paper since it was a SODA paper.} \NR{Oh yes.}

\subsection{Our Results}
We study fair division instances with agents having additive valuations over a set of indivisible items. The aim of this work is to push our understanding of the compatibility between two different classes of fairness notions: $\efx/\ef1$ with $\mms$ guarantees. 
Our main contribution is developing (simple) algorithms for achieving $\efx/\ef1$ and $\mms$ guarantees simultaneously. \\

\noindent
\textbf{Main Theorem:} For any fair division instance, we show that there exists
\begin{enumerate}
    \item a partial allocation that is both $2/3$-$\mms$ and $\efx$ [see Theorem~\ref{thm:efx+MMS-1} and Algorithm~\ref{alg:approx-mms+efx}].
    \item a complete allocation that is both $2/3$-$\mms$ and $\ef1$ [see Theorem~\ref{thm:ef1+MMS} and Algorithm~\ref{alg:approx-mms+ef1}].
\end{enumerate}
If we relax $2/3$-$\mms$ to $(2/3-\varepsilon)$-$\mms$ for any arbitrary constant $\varepsilon>0$, then the above allocations can be computed in pseudo-polynomial time. If in addition to that, we relax $\efx$/$\ef1$ to $(1-\delta)$-$\efx$/$(1-\delta)$-$\ef1$, then the allocations can be computed in polynomial time.

We note that the above results have led to a new approach for finding desired partial $\efx$ allocations, in particular, where we have a good bound on the amount of value each agent receives. It is known that $\efx$ is not compatible with the economic efficiency notion of Pareto optimality \citep{plaut2020almost}. Therefore, it may seem that, in order to guarantee $\efx$, one might have to sacrifice a lot of utility and agents may not receive bundles with high valuations. Nevertheless, using Algorithm~\ref{alg:approx-mms+efx}, we prove that we can still guarantee their $2/3$-$\mms$ value to every agent while finding a partial $\efx$ allocation. 

We use Algorithm~\ref{alg:approx-mms} developed by \citet{amanatidis2021maximum} to compute $2/3$-$\mms$ allocations as a starting point to have share-based guarantee. Here, as soon as an agent receives a bundle, she is taken out of consideration. This feature of the algorithm is incompatible with achieving any envy-based guarantees.\footnote{We note that this feature is common to many other algorithms achieving share-based guarantees in the fair division literature. See ``valid reductions'' in Section~\ref{sec:mms}.} We overcome this barrier and develop a novel algorithm (Algorithm~\ref{alg:approx-mms+efx}) that removes the \emph{myopic} nature of Algorithm~\ref{alg:approx-mms} and also looks into the future and modifies the already-allocated bundles if needed. Interestingly enough, the share-based guarantee that we maintain for a subset of agents (whose size keep growing) throughout the execution of Algorithm~\ref{alg:approx-mms+efx} helps us to prove envy-based guarantees as well. 

Our first result improves the guarantees shown by \citet{chaudhury2021little} where they develop a pseudo-polynomial time algorithm to compute a partial allocation that is both $1/2$-$\mms$ and $\efx$. Also, \citet{amanatidis2020multiple} develop an efficient algorithm to compute a complete allocation that is simultaneously $0.553$-$\mms$ and $0.618$-$\efx$; note that, this is incomparable to the guarantees that we develop in this work. On the other hand, the best known approximation factors, prior to our work, for simultaneous guarantees on $\mms$ and $\ef1$ was by \citet{amanatidis2020multiple} where they efficiently find allocations that are $4/7$-$\mms$ and $\ef1$.

%We use the algorithm developed by \citet{amanatidis2021maximum} to compute $2/3$-$\mms$ allocations and modify it in a way to achieve simultaneous $\efx$ guarantees as well (see Algorithm~\ref{alg:approx-mms+efx}). 

%One of the notable features of our algorithm (Algorithm~\ref{alg:approx-mms+efx}) is its simplicity. 

\iffalse
Recently, \citet{AR24eefx} stated an interesting open problem of whether $\eefx$ and $\ef1$ properties can simultaneously be achieved. While this problem still remains open, as a by-product of our above result, we prove the existence of
\begin{enumerate}
    \setcounter{enumi}{2}
    \item  a complete allocation that is both $2/3$-$\eefx$ and $\ef1$ [Theorem~\ref{thm:eefx+ef1}].
\end{enumerate}
 Again, for relaxations of $2/3$-$\eefx$ to $(2/3-\varepsilon)$-$\eefx$ and also $\ef1$ to $(1-\delta)$, we get pseudo-polynomial and polynomial run-time guarantees respectively. 
\fi 
 Finally, we also exhibit a constructive proof of
\begin{enumerate}
    \setcounter{enumi}{2}
    \item the existence of a (partial) allocation that is both $\alpha$-$\mms$ and $\efx$ for $\alpha=\max(2/3, \frac{1}{2-p/n})$, where $p < n$ goods are unallocated and given to charity such that no agent envies the charity [Theorem \ref{thm:efx+MMS-2}].
\end{enumerate}
Here, we improve the result of \citet{chaudhury2021little} where they prove the same existential result except that $\alpha=\frac{1}{2-p/n}$. We note that their (and our) algorithm has no power on what $p$ will be except that it cannot be larger than $n-1$. Hence, their result does not prove any existential result on the simultaneous guarantees for $\efx$ and $\delta$-$\mms$ for any constant $\delta>1/2$.%, while we efficiently compute, for any arbitrary constant $\delta>0$ and $\varepsilon>0$, a (partial) $(1-\delta)$-$\efx$ allocation with $(2/3-\varepsilon)$-$\mms$. %Their algorithm can be run in polynomial time at the cost of having $(1-\varepsilon)$-$\efx$ guarantee instead of exact $\efx$. 

%\todo{There is this criticism towards $\efx$ that it is not efficient since it is not feasible with PO. Thus, the worry is in order to guarantee $\efx$, one needs to sacrifice a lot of utility and agents don't receive bundles with high utility. However, in this work we prove that the agents can still receive $2/3$-$\mms$ while the partial allocation is $\efx$. The best guarantee for approx $\mms$ together with $\efx$ was only $1/2$ \citep{chaudhury2021little}.} 

%With an exception of $\ef1$ allocations, which is known to exist and can be computed in polynomial time, the notions of $\efx$ and $\mms$ have produced various interesting (and challenging) research directions. For instance, the biggest open problem in fair division is to resolve the existence of $\efx$ allocations for instances with four or more agents \cite{}. A recent work has proved that partial $\efx$ allocations can be computed in polynomial time where at most $n-1$ goods go to charity (and therefore remain unallocated) such that no agent envies the charity bundle \cite{}. On the other hand, even though it is now known that $\mms$ allocations may not always exist \cite{procaccia2015cake}, a significant amount of research has been focused on achieving better approximation guarantees for maximin share. For fair division instances, finds an allocation where every agent receives a bundle that is valued at least $2/3$ of her maximin share value \cite{}.  

\subsection{Further Related Work}\label{sec:related}
%The notion of maximin share is closely related to the popular max-min objective or the classic Santa Claus problem \cite{}. For instances with agents having identical valuations, an exact $\mms$ allocation exists, and here, finding an $\alpha$-$\mms$ allocation is equivalent to $\alpha$-approximation of the Santa Claus problem. The current best approximation factor known for the max-min objective under additive valuations is $\Tilde{\mathcal{O}}(m^{\varepsilon})$ for any $\varepsilon>0$ \cite{}. \NR{I have taken the above paragraph from your soda paper. It would be good to rephrase this part.} \HA{Do we even need it?}

For the $\mms$ problem, \cite{kurokawa2018fair} showed the existence of $2/3$-$\mms$ allocations, while \citeauthor{barman2020approximation} established its tractability for instances with additive valuations. Many follow-up works are filled with extensive studies to improve the approximation factor for $\mms$ allocations e.g., see~\cite{amanatidis2017approximation,kurokawa2018fair,ghodsi2018fair,barman2020approximation,garg2020improved,FST21,simple,akrami2024breaking} for additive,~\cite{barman2020approximation, ghodsi2018fair,uziahu2023fair} for submodular,~\cite{ghodsi2018fair,seddighin2022improved,MMS-XOS} for XOS, and~\cite{ghodsi2018fair, seddighin2022improved} for subadditive valuations. %The current best result for existence of $\alpha$-$\mms$ allocations for additive valuations is for $\alpha=(\frac{3}{4}+\frac{3}{3836})$ \cite{}. 

%Positive results on $\efx$ are known for very restricted settings. 
%The journey of $\efx$ allocations is filled with equally extensive studies to push the limits of our understanding. 

% Over the past few years, the study of $\efx$ allocations has led to development of various new techniques and connections with graph theory via \emph{rainbow cycle number} \cite{}. 

For $\efx$ allocations, \cite{plaut2020almost} proved its existence for two agents with monotone valuations. A breakthrough result by \citeauthor{chaudhury2020efx} proved the existence of $\efx$ allocations for instances with three agents having additive valuations. Many follow-up works strengthened this result with more general valuations ~\citep{chaudhury2020efx,berger2021almost,AkramiACGMM23}. $\efx$ allocations exist when agents have identical~\citep{plaut2020almost}, binary~\citep{halpern2020fair}, or bi-valued~\citep{amanatidis2021maximum} valuations.
Several approximations~\citep{chaudhury2021little,amanatidis2020multiple,chan2019maximin,farhadi2021almost} and relaxations~\citep{amanatidis2021maximum,caragiannis2019envy,berger2021almost,mahara2021extension,chasmjahan23,aram22,ef2x} of $\efx$ have become an important line of research in discrete fair division. Inspired by the work of \citet{ABCGL18}, a recent work by \citet{Caragiannis2023} has defined an interesting and useful relaxation of $\efx$, called \emph{epistemic $\efx$} ($\eefx$). $\eefx$ allocations are guaranteed to exist for instances with monotone valuations \cite{AR24eefx} and can be computed in polynomial time for instances with additive valuations \citep{Caragiannis2023}.

Some of these fairness criteria have also been studied in combination with other objectives, such as Pareto optimality \citep{barman2018finding}, truthfulness \citep{amanatidis2016truthful,amanatidis2017truthful} or maximizing the Nash welfare \citep{caragiannis2016unreasonable,caragiannis2019envy,chaudhury2021little}.

%%%I have removed the part about PROP1 and PROPX
%Proportionality up to one good (PROP1)~\citep{conitzer2017fair} is another relaxation of proportionality which can be guaranteed together with Pareto optimality~\citep{barman2019proximity}. Proportionality up to any good (PROPX) on the other hand, is not a feasible notion in the goods setting~\citep{aziz2020polynomial}.

 %Unfortunately, it seems that $\ef$1 has moved way too far and has lost the fairness properties of envy-freeness.

An excellent recent survey by~\citet{survey2022} discusses the above fairness concepts and many more.
Another aspect of discrete fair division which has garnered an extensive research is when the items that needs to be divided are \emph{chores}. We refer the readers to the survey by~\citet{guo2023survey} for a comprehensive discussion.

\section{Definitions and Notation} \label{sec:prelim}
For any positive integer $k$, we use $[k]$ to denote the set $\{1,2,\ldots,k\}$.
We write $\N=[n]$ to denote the set of $n$ agents and $\M=\{g_1, \ldots, g_m\}$ to denote the set of $m$ indivisible items. For an agent $i \in \N$, the valuation function $v_i:2^\M \rightarrow \R_{\geq 0}$ represents her value over the set of items. For simplicity, we will often write $g$ instead of $\{g\}$ for an item $g \in \M$. In this work, we assume that valuation functions $v_i$'s are additive i.e., for any agent $i \in \N$, $v_i(S) = \sum_{g \in S} v_i(g)$ for any subset $S \subseteq \M$. We denote a fair division instance by $\I = (\N, \M, \V)$, where $\V=(v_1, v_2, \ldots, v_n)$. When we say `fair division instance with additive valuations', we mean an instance with every agent having an additive valuation.

An allocation $X = (X_1, X_2, \ldots, X_n)$ is a partition of a subset of $\M$ into $n$ bundles, such that $X_i$ is the bundle allocated to agent $i \in [n]$ and $P(X) = M \setminus \bigcup_{i \in [n]} X_i$ is the set (pool) of unallocated goods. If $P(X)=\emptyset$, then we say $X$ is a \emph{complete} allocation, otherwise, we say $X$ is a \emph{partial} allocation. Also, we write $\Pi_n$ to denote the set of all partitions of $\M$ into $n$ bundles, i.e., the set of all $n$-partitions of the set $\M$.

We study the approximations and relaxations of classic fairness notions of envy-freeness and proportionality. We begin by defining the concept of \emph{strong envy} to state the fairness notions of $\efx$ and $\ef1$.

\begin{definition}[Strong Envy]
    Upon receiving bundle $A$, we say that agent $i$ \emph{strongly envies} a bundle $B$, if there exists an item $g \in B$ such that $v_i(A)<v_i(B \setminus g)$. Given an allocation $X$, agent $i$ \emph{strongly envies} agent $j$ if there exists an item $g \in X_j$ such that $v_i(X_i)<v_i(X_j \setminus g)$.
\end{definition}

\begin{definition}[Envy-freeness up to any item ($\efx$)]
    An allocation $X$ is $\efx$ if no agent strongly envies any other agent. In other words, for all $i,j \in \N$ and all $g \in X_j$, we have $v_i(X_i) \geq v_i(X_j \setminus g)$.
\end{definition}

For any $\alpha \geq 0$, an allocation $X$ is $\alpha$-$\efx$, if for all $i,j \in \N$ and all $g \in X_j$, we have $v_i(X_i) \geq \alpha \cdot v_i(X_j \setminus g)$.

\begin{definition}[Envy-freeness up to one item ($\ef1$)]
    An allocation $X$ is $\ef1$, if for all $i,j \in \N$, we either have $v_i(X_i) \geq v_i(X_j)$, or there exists $g \in X_j$ such that $v_i(X_i) \geq v_i(X_j \setminus g)$.
\end{definition}

Similarly, for any $\alpha \geq 0$, an allocation $X$ is $\alpha$-$\ef1$, if for all $i,j \in \N$, we either have $v_i(X_i) \geq \alpha \cdot v_i(X_j)$ or there exists $g \in X_j$ such that $v_i(X_i) \geq \alpha \cdot v_i(X_j \setminus g)$. 

We now define the concept of \emph{most envious agent} for a bundle, which come useful in developing  our algorithms.
\begin{definition}[Most envious agent]
    Given a bundle $B$ and a (partial) allocation $X$, an agent $i \in \N$ is a \emph{most envious agent} of bundle $B$, if there exists a proper subset $B' \subsetneq B$ such that $v_i(B')>v_i(X_i)$ and no other agent $j \in \N$ such that $j \neq i$ strongly envies $B'$.    
\end{definition}

\begin{restatable}{observation}{obs}
    Given a fair division instance with additive valuations, consider a bundle $B$ and a (partial) allocation $X$. If there exists an agent $i$ who strongly envies $B$, then there exists an agent who is a most envious agent of $B$, and she can be identified in polynomial time.
\end{restatable}

\begin{proof}
    For all agents $j$ who strongly envy $B$, let $B_j \subset B$ be an inclusion-wise minimal subset such that $v_j(X_j)<v_j(B)$. Let $j^*$ be such that $|B_{j^*}|$ is minimum. Then no agent $j$ strongly envies $B_{j^*}$ and thus $j^*$ is a most envious agent of $B$. The sets $B_j$'s can be computed in polynomial time for each agent $j$ by greedily removing goods from $B$ as long as the value of the remaining set exceeds $v_j(X_j)$. And therefore, agent $j^*$ can be identified in polynomial time as well.
\end{proof}

We next discuss the share-based fairness notion of 
\emph{maximin share} ($\mms$). We define the maximin share value of an agent $i \in [n]$ as the maximum value she can guarantee for herself, if she partitions the goods into $n$ bundles and receives a bundle with minimum value (to her). Then, for any agent $i \in [n]$, we write her maximin share value as,
\begin{align}
    \mms^n_i(\M) \coloneq \max_{(A_1, \ldots, A_n) \in \Pi_n} \min_{A_j} v_i(A_j).
\end{align}
where, $\Pi_n$ is the set of all partitions of $\M$ into $n$ bundles. When $n$ and $\M$ are clear from the context, we write $\mms_i$ instead of $\mms^n_i(\M)$.

\begin{definition}[$\alpha$-$\mms$ Allocation]
    For any $\alpha \in [0,1]$, allocation $X$ is $\alpha$-$\mms$, if for all agents $i \in \N$, we have $v_i(X_i) \geq \alpha\cdot\mms_i$. We say an allocation $X$ is $\mms$, if it is $1$-$\mms$.
\end{definition}

Note that, the definition of $\mms$ dictates that for all $i \in \N$, there exists a partition $(A_1, \ldots, A_n)$ of $\M$ such that $v_i(A_j) \geq \mms^n_i(\M)$ for all $j \in [n]$. We call such a partition as an \emph{$\mms$-partition of agent $i$}. Similarly, a partition $(A_1, \ldots, A_n)$ of $\M$ such that $v_i(A_j) \geq \alpha\mms^n_i(\M)$ for all $j \in [n]$ is an $\alpha$-$\mms$-\emph{partition of agent $i$}.

\begin{proposition}[\citet{woeginger1997polynomial}]\label{prop:ptas}
    Given any fair division instance with additive valuations, there exists a PTAS to compute an $\mms$-partition of any agent $i \in \N$, and hence her $\mms_i$ value as well.
\end{proposition}

Lastly, we define two graphs inspired by share-based and envy-based fairness notions, that will prove useful in our algorithms.

\begin{definition}[Threshold-Graph]
    Given a partition $Y=(Y_1, \ldots, Y_n)$ of $\M$ into $n$ bundles and given a vector $t=(t_1, \ldots, t_n) \in \R^n_{\geq 0}$, we define the \emph{threshold-graph} as an undirected bipartite graph $T_{\langle Y,t \rangle}=(V,E)$, where $V$ has one part consisting of $n$ nodes corresponding to the agents and another part with $n$ nodes corresponding to the bundles $Y_1, \ldots, Y_n$. There exists an edge $(i,j)$ between (the node corresponding to) agent $i$ and (the node corresponding to) bundle $Y_j$ if and only if $v_i(Y_j) \geq t_i$. For all $i \in [n]$, we call $t_i$, the threshold share value of agent $i$.

     For a subset $S$ of the nodes, we write $N(S)$ to denote the set of neighbours of the nodes in $S$ in the threshold graph.
\end{definition}

\begin{definition}[Envy-Graph]
    Given an allocation $X$, we define the \emph{envy-graph} of $X$ as a directed graph $G_X = (V,E)$ where $V$ is a set of $n$ nodes corresponding to agents, and there exists an edge from (the node corresponding to) agent $i$ to (the node corresponding to) agent $j$, if and only if agent $i$ envies agent $j$, i.e.,  $v_i(X_j) > v_i(X_i)$. 
\end{definition}

\section{\boldmath Relations Between $\efx$/$\ef1$ and $\mms$} \label{sec:relation}
In this section, we briefly discuss the guarantees $\efx$/$\ef1$ allocations can provide for $\mms$ and vice versa. \citet{comparing-efx} gave a comprehensive comparison between these notions of fairness. Here we mention a few.
%In particular, they proved that while a complete $\efx$ allocation implies $4/7$-$\mms$ guarantee, there exists complete $\efx$ allocations which are as bad as $0.5914$-$\mms$. Nevertheless, these guarantees become relevant only when a complete $\efx$ allocation exists which, in itself, is a big open problem. 
\begin{proposition}[\citet{comparing-efx}]
    For arbitrary $n \geq 1$, any $\efx$ allocation is also a $4/7$-$\mms$ allocation. On the other hand, an $\efx$ allocation is not necessarily an $\alpha$-$\mms$ allocation for $\alpha>0.5914$ and large enough $n$.
\end{proposition}

%In this work, we show that the relation is even weaker between $\ef1$ and $\mms$. Furthermore, we show that, for any $\alpha>0$, $\alpha$-$\mms$ allocations do not guarantee any bounded approximation ratio for $\ef1$ and hence also $\efx$. 
%The proofs of following propositions can be found in Appendix~\ref{app:missing-proofs}.

\begin{restatable}[\citet{comparing-efx}]{proposition}{propone}
    An $\ef1$ allocation is not necessarily an $\alpha$-$\mms$ allocation for any $\alpha>1/n$.
\end{restatable}

\begin{proof}
    Consider the instance $\I$ with $n$ agents with identical valuation $v$ over $2n-1$ items. Assume $v(g_i)=1$ for all $i \in [n-1]$ and $v(g_i)=1/n$ for all $i \in [2n-1]\setminus[n-1]$. Consider the following partition of $\M$ into $n$ bundles. For $i \in [n-1]: A_i = \{g_i\}$ and $A_n=\{g_n, \ldots, g_{2n-1}\}$. $v(A_i)=1$ for all $i \in [n]$ and thus the $\mms$ value of all agents is $1$. Now consider the following allocation. $X_i = \{g_i,g_{n+i-1}\}$ for $i \in [n-1]$ and $X_n = \{g_{2n-1}\}$. For all $i \in [n-1]$ $v(X_i)=1+1/n$ and $v(X_n)=1/n$. The allocation $X$ is $\ef1$ since $v(X_n) \geq v(X_i \setminus \{g_i\})$ for all $i \in [n-1]$. However, $X$ is $1/n$-$\mms$.
\end{proof}

\begin{restatable}{proposition}{proptwo}
    For $n \geq 3$ and any $\alpha>0$, an $\alpha$-$\mms$ allocation is not necessarily $\beta$-$\ef1$ for any $\beta > 0$.
\end{restatable}

\begin{proof}
    Consider the instance $\I$ with $n \geq 3$ agents with identical valuation $v$ over $2$ items with $v(g_1) = v(g_2)=1$. Clearly the $\mms$ value of all the agents is $0$ and thus all allocations are $\alpha$-$\mms$ for any $\alpha > 0$. Now consider the allocation $X$ with $X_n = \{g_1, g_2\}$ and $X_i=\emptyset$ for all $i \in [n-1]$. For all $i \in [n-1]$ and $\beta>0$, $v(X_i) < \beta v(X_n \setminus \{g_1\})$. Thus, $X$ is not $\beta$-$\ef1$.
\end{proof}

Therefore, we can conclude that, by guaranteeing one of approximate $\mms$ or approximate $\efx$/$\ef1$, one cannot obtain a good guarantee for the other notion of fairness for free.

%In Section \ref{sec:limitations} \todo{...}, we discuss other existing algorithms to compute either of $2/3$-$\mms$ or $\efx$/$\ef1$, where none of these algorithms guarantee the other fairness criterion. This suggests that while guaranteeing each of these fairness notions (besides $\ef1$) is difficult on its own, satisfying both kinds of fairness guarantees simultaneously becomes an even more challenging problem.
\section{\boldmath Guaranteeing $\frac{2}{3}$-$\mms$}\label{sec:mms}

In this section, we describe and analyze the algorithm developed by \citet{amanatidis2017approximation} to compute $2/3$-$\mms$ allocations for fair division instances with additive valuations. We rewrite it and analyze it in our own words (in Algorithm~\ref{alg:approx-mms}) since we use it to develop our main algorithm (Algorithm~\ref{alg:approx-mms+efx}) to compute allocations that are both $2/3$-$\mms$ and $\efx$. 

\citeauthor{budish2011combinatorial}, while introducing maximin share, also showed that it is scale-invariant. That is, we can assume $\mms_i=1$ for all agents $i \in \N$. 

Surprisingly, Algorithm~\ref{alg:approx-mms} does \emph{not} rely on the two most commonly used tools for computing approximate $\mms$ allocations, namely \emph{ordered instances} and \emph{valid reductions}. 

\paragraph{Ordered Instances.} We say that a fair division instance is \emph{ordered} if we can rename the goods such that $v_i(g_1) \geq \ldots \geq v_i(g_m)$ holds true for all agents $i \in [n]$. \citet{barman2020approximation} showed that, for any $\alpha \in [0,1]$, if $\alpha$-$\mms$ allocations exist for ordered instances, then $\alpha$-$\mms$ allocations exist for all instances. Therefore, it becomes natural to assume ordered instances while studying $\alpha$-$\mms$ allocations, as done by many of the previous works in the literature \cite{garg2019approximating,garg2021improved,simple,akrami2024breaking}.

\paragraph{Valid Reductions.}\label{par:valid-reduction} Another common tool in computing $\alpha$-$\mms$ allocations, is valid reductions used in multiple works \cite{bouveret2016characterizing,kurokawa2016can,kurokawa2018fair,amanatidis2017approximation,ghodsi2018fair,garg2019approximating,garg2021improved,simple,akrami2024breaking}. Allocating a bundle $B$ to an agent $i$ and removing agent $i$ with $B$ from consideration to \emph{reduce} to a smaller instance is said to be \emph{valid}  if
\begin{itemize}
    \item $v_i(B) \geq \alpha \cdot \mms_i$, and
    \item $\mms^{n-1}_j(\M \setminus B) \geq \mms^n_j(\M)$ for all $j \neq i$.
\end{itemize}
In other words, allocating $B$ to agent $i$ is harmless since, not only $B$ satisfies $i$, but also removing $i$ and $B$ from the instance, it does not decrease the $\mms$ value of the remaining agents. Hence, one can, without loss of generality, allocate $B$ to $i$ and proceed with computing an $\alpha$-$\mms$ allocation in the reduced instance. 

Unfortunately, none of these tools can be used when dealing with envy-based notions of fairness. And hence, most of the previous works that achieve approximate $\mms$ guarantees do not obtain any envy-based criteria results. On the other hand, most of the previous work that achieve simultaneous guarantees for $\mms$ and $\efx$/$\ef1$ are obtained by manipulating algorithms that provide $\efx$/$\ef1$ guarantees so that some approximation for $\mms$ can also be achieved \citep{chaudhury2021little,amanatidis2020multiple}. However, so far, the envy-based algorithmic techniques have not been strong enough to also attain $2/3$-$\mms$ guarantee.

%\todo{cite the papers using this technique also known as lone-divider.}

\begin{algorithm}[t]
    \caption{$\mathtt{approxMMS}(\I)$}\label{alg:approx-mms}
    \textbf{Input:} A fair division instance $\I = (\N, \M, \V)$ with additive valuations\\
    \textbf{Output:} An allocation $X$
    \begin{algorithmic}[1]
        \State Let $\mms_i = \mms^n_i(\M)$ for all $i \in [n]$
        \While{$\N \neq \emptyset$}
            \State $n \leftarrow |\N|$
            \State Let $i \in \N$
            \State Let $(X_1, \ldots, X_n)$ be a partition of $\M$ such that $v_i(X_j) \geq \frac{2}{3} \mms_i$
            \State Let $T_{\langle X,t \rangle}$ be the threshold-graph with $X=(X_1, \ldots, X_n)$ and $t=\frac{2}{3}(\mms_1, \ldots, \mms_n)$ for agents in $[n]$
            \State Let $M=\{(k+1,X_{k+1}), \ldots, (n,X_n)\}$ be a matching of size at least $1$ such that $N(\{X_{k+1}, \ldots, X_n\})=\{k+1, \ldots, n\}$ and $X_j$ is matched to $j$ for all $j \in [n]\setminus[k]$; \label{step:matching}
            \State $\N \leftarrow [k]$;
            \State $\M \leftarrow \M \setminus \bigcup_{\ell \in [n] \setminus [k]} X_\ell$;
        \EndWhile 
        \Return $(X_1, X_2, \ldots, X_n)$;
    \end{algorithmic}
\end{algorithm}

\noindent
\textbf{Overview of Algorithm~\ref{alg:approx-mms}:} Algorithm~\ref{alg:approx-mms} successively allocates a bundle of items to some selected agents in each step and removes them from consideration. In particular, in each round of Algorithm~\ref{alg:approx-mms} with $n'$ remaining agents, we ask a remaining agent $i$ to divide the remaining items into $n'$ bundles $X_1, \ldots, X_{n'}$, each of value at least $2/3$ to her. We prove, in Lemma~\ref{lem:reduction}, that the above is always possible at every step of the algorithm. Then, we consider the threshold graph $T_{\langle X, t \rangle}$ with $X=(X_1, \ldots, X_{n'})$ and $t=(\frac{2}{3}, \ldots, \frac{2}{3})$ and find a matching between the bundles and the agents such that (i) every matched agent has a value of at least $2/3$ for the bundle matched to her and (ii) every unmatched agent values any of the matched bundles at less than $2/3$. We allocate according to this matching, and remove the matched agents with their matched bundles. As long as there is any remaining agent, we repeat the above process. See Algorithm~\ref{alg:approx-mms} for the pseudo-code of this algorithm. A similar technique is also used in \cite{steinhaus1948problem,Kuhn,ordinalOne,hummel}.

\begin{restatable}[\citet{amanatidis2017approximation}]{theorem}{twothirdmms}\label{thm:2-3-mms}
    For fair division instances with additive valuations, Algorithm~\ref{alg:approx-mms} returns a $\frac{2}{3}$-$\mms$ allocation. 
\end{restatable}

First we give a very simple proof for Lemma \ref{lem:reduction} which was proven by \citet{PW2014} in a complicated fashion.
\begin{restatable}{lemma}{reduction}\label{lem:reduction}
    Fix an agent $i \in \N$ and some $k<n$. Consider $k$ bundles $A_1, \ldots, A_k \subseteq \M$ such that for all $j \in [k]$, we have $v_i(A_j) < \frac{2}{3}\cdot\mms^n_i(\M)$ for agent $i$. Then, there exists a partition $(B_1, \ldots, B_{n-k})$ of the remaining items in $\M \setminus \cup_{j \in [k]}A_j$ into $n-k$ bundles such that $v_i(B_j) \geq \frac{2}{3} \cdot \mms^n_i(\M)$ for all $j \in [n-k]$.
\end{restatable}
  \paragraph{Proof idea.} We first give an overview of our idea before formally proving Lemma~\ref{lem:reduction}. Let us begin by fixing an $\mms$-partition of agent $i$ in the given instance. Then, given the $k$ bundles $A_1, \ldots, A_k \subseteq \M$ with the property as stated in Lemma~\ref{lem:reduction}, we categorize the bundles of this $\mms$-partition depending on how much value these bundles have after the removal of the items in $A_1 \cup \ldots \cup A_k$. All the bundles with remaining value of at least $2/3$ are in the set $C^0$ and all the bundles with remaining value at least $1/3$ and at most $2/3$ are in the set $C^1$. By pairing the bundles in $C^1$ and merging the items in this pair, note that, we manage to create $\lfloor \frac{|C^1|}{2} \rfloor$-many bundles of value at least $2/3$. Moreover, all the bundles in $C^0$ are of value at least $2/3$. So, proving that $|C^0| + \lfloor \frac{|C^1|}{2} \rfloor \geq n-k$ will suffice. We show it by upper-bounding the value of the removed items for agent $i$.
\begin{proof}
    For a given fair division instance $\I$ and an agent $i \in \N$, let $(C_1, \ldots, C_n)$ be an $\mms$-partition of $i$. Let $A = A_1 \cup \ldots \cup A_k$ and write $D_j = C_j \cap A$ and $C'_j = C_j \setminus A$ for all $j \in [n]$. We define $C^0$, $C^1$, and $C^2$ as following depending on the amount of value removed from $C_j$'s after the removal of the items in $A$. 
    \begin{itemize}
        \item $C^0 = \{j \in [n] \mid v_i(D_j) \leq \frac{1}{3}\}$.
        \item $C^1 = \{j \in [n] \mid \frac{1}{3} < v_i(D_j) \leq \frac{2}{3}\}$.
        \item $C^2 = \{j \in [n] \mid \frac{2}{3} < v_i(D_j)\}$.
    \end{itemize}
    Let $n_0 = |C^0|$, $n_1 = |C^1|$, and $n_2 = |C^2|$. Note that $n = n_0 + n_1 + n_2$. Without loss of generality, we assume that $C^0 = \{1, \ldots, n_0\}$, $C^1 = \{n_0+1, \ldots, n_0 + n_1\}$, and $C^2 = \{n_0+n_1+1, \ldots, n\}$. We aim to create $n-k$ many bundles $B_1, \dots, B_{n-k} \subset \M \setminus A$, each of value at least $2/3$ to agent $i$. To begin with, we set $B_j=C'_j$ for all $j \in [n_0]$. Next, we pair the bundles with indices in $C^1$ and merge the items in them. Formally, for all $j \in [\lfloor \frac{n_1}{2} \rfloor]$, we set $B_{n_0+j} = C'_{n_0+2j-1} \cup C'_{n_0+2j}$. 

    First, note that, for all $j \in C^0$, $v_i(B_j)=v_i(C'_j) \geq \frac{2}{3}$. Also, since for all $j \in C^1$, $v_i(C'_j) \geq \frac{1}{3}$, for all $\ell \in [n_0 + \lfloor \frac{n_1}{2} \rfloor]$, $v_i(B_j) \geq \frac{2}{3}$. That is, we have $v_i(B_j) \geq 2/3$ for all $j \in [n_0] \cup [n_1]$. Therefore, to complete our proof, it suffices to establish that $n_0 + \lfloor \frac{n_1}{2} \rfloor \geq n-k$.  We have,
    \begin{align*}
        \frac{2}{3} \cdot k &> \sum_{j \in [k]} v_i(A_j) \tag{$v_i(A_j) < 2/3$ for all $j \in [k]$} \\
        &= \sum_{j \in [n]} v_i(C_j \cap A) = \sum_{j \in [n]} v_i(D_j) \\
        &= \sum_{j \in [n_0]} v_i(D_j) + \sum_{j \in [n_1]\setminus[n_0]} v_i(D_j) + \sum_{j \in [n_2]\setminus[n_1]} v_i(D_j)\\
        &\geq \frac{1}{3} n_1 + \frac{2}{3} n_2.
    \end{align*}
    That is, $k > \frac{n_1}{2} + n_2$, or equivalently, $n-k < n - (\frac{n_1}{2} + n_2) = n_0 + \frac{n_1}{2}$. Therefore, we have  $n_0 + \lfloor \frac{n_1}{2} \rfloor \geq n-k$, as desired.
\end{proof}

Given a set of nodes $S$ in a threshold graph $T$, $N(S)$ is the set of all neighbors of the nodes in $S$.

\begin{restatable}{lemma}{matching}\label{lem:matching}
    For a given partition $X=(X_1, \ldots, X_n)$ of $\M$ and a threshold vector $t=(t_1, \ldots, t_n)$, assume there is an agent $i$ such that $v_i(X_j) \geq t_i$ for all $j \in [n]$. Then $T_{\langle X,t \rangle}$ has a non-empty matching $M=\{(i_1, X_{j_1}), \ldots, (i_k, X_{j_k})\}$ such that $N(\{X_{j_1}, \ldots, X_{j_k}\})=\{i_1, \ldots, i_k\}$ which can also be computed in polynomial time. 
\end{restatable}
\begin{proof}
    First we compute a maximum matching $M^*$ in $T_{\langle X,t \rangle}$ (which can be done in polynomial time).
    If $M^*$ is a perfect matching between $[n]$ and $(X_1, \ldots, X_{n})$, then clearly the lemma holds. Otherwise, there must exists a Hall's violator set $S \subset \{X_1, \ldots, X_n\}$ with $|N(S)| < |S|$ which can be computed in polynomial time. Compute a minimal such set $S$ by removing elements one by one if after the removal still the inequality ($|N(S)| < |S|$) holds. Note that $i \in N(S)$ and hence, $|S| \geq 2$. By minimality of $S$, the Hall's condition holds for any proper subset of $S$. Let $T \subset S$ and $|T| = |S|-1$. We have $|S|-1 = |T| \leq |N(T)| \leq |N(S)| < |S|$. Hence $|N(T)| = N(S) = |S|-1$. Since the Halls' condition holds for $T$, there exists a matching $M$ covering $N(T)$. Since $N(T) \subseteq N(S)$ and $|N(T)|=|N(S)|$, we have that $M$ is covering $N(S)$. Therefore, there exists no edge between the agents outside $M$ and the bundles inside $M$. 
\end{proof}

\paragraph{Proof of Theorem~\ref{thm:2-3-mms}}
Now we are ready to prove \ref{thm:2-3-mms}. To begin with, note that, if Algorithm~\ref{alg:approx-mms} terminates, each agent $i \in [n]$ is matched to (and allocated) say, bundle $X_i$ in some threshold-graph with the threshold of agent $i$ being $\frac{2}{3} \cdot \mms_i$, i.e., $v_i(X_) \geq \frac{2}{3} \cdot \mms_i$. Thus, if the algorithm terminates, it must return a $\frac{2}{3} \mms$ allocation. 

    Now, in order to prove the termination of Algorithm \ref{alg:approx-mms}, we prove that in each iteration of the while-loop, the set-size $|\N|$ of the remaining agents strictly decreases. Therefore, the while-loop can iterate for at most $n$ many times and hence Algorithm~\ref{alg:approx-mms} terminates. 
    
    Consider an arbitrary iteration of the while-loop in Algorithm~\ref{alg:approx-mms}. Let $\N$ and $\M$ be the initial set of agents and items respectively and let $\N'=[n']$ and $\M'$ be the set of remaining agents and items respectively in the beginning of this iteration of the while-loop. If $\N' = \emptyset$, then the algorithm obviously terminates. 
    
    Let us now consider some agent $i \in \N'$. First, we prove that there exists a partition $(X_1, \ldots, X_{n'})$ of $\M'$ such that $v_i(X_j) \geq \frac{2}{3} \cdot \mms_i$ for all $j \in [n']$. We write $X_{n'+1}, \ldots, X_n$ to denote the bundles that are matched to agents $n'+1, \ldots, n$ in the previous iterations of the while-loop. Furthermore, by the choice of matching $M$ in Step~\ref{step:matching} of Algorithm~\ref{alg:approx-mms}, agent $i$ did not have an edge (in the then threshold graphs) to any of these bundles $X_{n'+1}, \ldots, X_n$. That is, $v_i(X_j) < \frac{2}{3} \cdot \mms_i$ for all $j \in [n] \setminus [n']$. By Lemma \ref{lem:reduction}, there exists a partitioning $(X_1, \ldots, X_{n'})$ of $\M'$ such that $v_i(X_j) \geq \frac{2}{3}\mms_i$ for all $j \in [n']$.

    By Lemma \ref{lem:matching}, there exists a matching $M$ with no edge between the agents outside $M$ and the bundles inside $M$. Now without loss of generality, by renaming the agents and bundles, assume $M=\{(k+1, X_{k+1}), \ldots, (n', X_{n'}))\}$. Since we know $M \neq \emptyset$, $k < n'$ and thus the number of remaining agents decreases in the end of this iteration of the while-loop. \qed

Using Proposition \ref{prop:ptas}, it is easy to see that there exists a PTAS to compute the partition $(B_1, \ldots, B_{n-k})$ in Lemma \ref{lem:reduction}. Formally, the following lemma holds.

\begin{lemma}\label{lem:reduction-poly}
Fix an agent $i \in \N$ and some $k<n$. Consider $k$ bundles $A_1, \ldots, A_k \subseteq \M$ such that for all $j \in [k]$, we have $v_i(A_j) < \frac{2}{3}\cdot\mms^n_i(\M)$ for agent $i$. Then, a partition $(B_1, \ldots, B_{n-k})$ of the remaining items in $\M \setminus \cup_{j \in [k]}A_j$ into $n-k$ bundles can be computed in polynomial time such that $v_i(B_j) \geq (\frac{2}{3}-\varepsilon) \cdot \mms^n_i(\M)$ for all $j \in [n-k]$ and all constant $\varepsilon>0$.
\end{lemma}
\begin{proof}
    In the proof of Lemma \ref{lem:reduction}, if an $\mms$-partition of $i$ is given, computing the bundles in $C^j$ for all $j \in \{0,1,2\}$ and consequently obtaining $(B_1, \ldots, B_{n-k})$ can be done in polynomial time. Fix a constant $\varepsilon>0$ and an agent $i$. By Proposition \ref{prop:ptas}, a $(1-\varepsilon/2)$-$\mms$ partition of $i$ can be computed in polynomial time and following the same arguments, a partition $(B_1, \ldots, B_{n-k})$ of the remaining items into $n-k$ bundles can be computed in polynomial time such that $v_i(B_j) \geq (\frac{2}{3}-\varepsilon) \cdot \mms^n_i(\M)$.
\end{proof}

\section{\boldmath $\frac{2}{3}$-$\mms$ Together with $\efx$}\label{sec:main-1}
%\todo{reword the first two paragraphs.}

In this section, we modify Algorithm \ref{alg:approx-mms} such that the output is a (partial) allocation which is still $2/3$-$\mms$ and now becomes $\efx$ as well.  Note that, in Algorithm \ref{alg:approx-mms} and generally in the algorithmic technique of \citet{amanatidis2017approximation}, once an agent receives a bundle $X_i$, $X_i$ does become her bundle in the final output allocation. So, once agent $i$ receives the bundle $X_i$, she is out of the consideration. This guarantees that agent $i$ will have the same utility $v_i(X_i)$ in the end of the algorithm but it does not guarantee anything about how much $i$ values other bundles formed once she is removed from consideration. Therefore, it cannot not guarantee $\efx$ (or even $\ef1$) property. 

We overcome this barrier by developing Algorithm \ref{alg:approx-mms+efx} in this section. Here, we again allocate a bundle of items to some selected agents in each step, but we modify them carefully in a later stage. As we will describe next, this feature of our algorithm removes the \emph{myopic} nature of Algorithm~\ref{alg:approx-mms} and lets us achieve envy-based fairness guarantees, while maintaining $2/3$-$\mms$ guarantees. \\

\noindent
\textbf{Overview of Algorithm~\ref{alg:approx-mms+efx}:} In each round of Algorithm~\ref{alg:approx-mms+efx} with $n' \leq n$ remaining agents, we ask a remaining agent $i$ to partition the remaining goods into $n'$ bundles $X_1, \ldots, X_{n'}$ of value at least $(2/3)\mms_i$. We prove, in Lemma~\ref{lem:reduction}, that it is always feasible to perform the above process at every step of the algorithm. We then shrink these bundles to guarantee that every remaining agent values each strict subset of these bundles less than $2/3$ fraction of their $\mms$ value. For simplicity, we rename the shrinked bundles again as $X_1, \ldots, X_{n'}$. 

Now, let us assume that, after the process of shrinking, we still have an agent $j$ who was allocated a bundle in previous iterations and who strongly envies one of $X_1, \ldots, X_{n'}$, say, for instance, $X_j$. Let us denote $a^*$ to be the most envious agent of $X_j$. We allocate, to $a^*$, a subset of $X_j$ which $a^*$ envies but no agent strongly envies. In this way, we guarantee two things at each point during the algorithm, the current (partial) allocation among the agents who received a bundle so far is (a) $\efx$ and (b) all these agents receive $2/3$ fraction of their $\mms$ value. See Algorithm \ref{alg:approx-mms+efx} for the pseudocode.

To the best of our knowledge, none of the previous algorithms computing an $\efx$ allocation allocates a bundle to some of the agents and nothing to the rest in an intermediate step. It might also seem counter-intuitive to do so, since we need to guarantee that there are enough items left to satisfy the agents who have received nothing so far. We are able to make it possible in Algorithm~\ref{alg:approx-mms+efx}, since we know that all the remaining agents (who have not yet received anything) value all the already allocated bundles less than $2/3$ fraction of their $\mms$ value. Interestingly enough, the share-based guarantee that we are maintaining helps us to prove envy-based guarantees as well.

\begin{algorithm}[t]
    \caption{$\mathtt{approxMMSandEFX}(\I)$}\label{alg:approx-mms+efx}
    \textbf{Input:} A fair division instance $\I = (\N, \M, \V)$ with additive valuations\\
    \textbf{Output:} An allocation $X$
    \begin{algorithmic}[1]
        \State Let $\mms_i = \mms^n_i(\M)$ for all $i \in [n]$
        \State $\N' \leftarrow [n]$
        \While{$\N' \neq \emptyset$}\label{while}
            \State $n' \leftarrow |\N'|$
            \State Let $i \in \N'$
            \State Let $(X_1, \ldots, X_{n'})$ be a partition of $\M$ such that $v_i(X_j) \geq \frac{2}{3} \mms_i$ \label{line:mms_partition}
            \For{$j \in [n']$}
                \State $X_j \leftarrow$ minimal subset of $X'_j \subseteq X_j$ such that $\exists i' \in [n']$ with $v_{i'}(X'_j) \geq \frac{2}{3}\mms_{i'}$ \label{line}
                \If{$\exists a \in [n]\setminus[n']$ such that $a$ strongly envies $X_j$}\label{if}
                    \State Let $a^* \in [n]\setminus[n']$ be a most envious agent of $X_j$ \label{line:choice}
                    \State Let $X'_j \subseteq X_j$ be minimal such that $v_{a^*}(X'_j)>v_{a^*}(X_{a^*})$ and no agent strongly envies $X'_j$ \label{line:a*}
                    \State $\M \leftarrow \M \cup X_{a^*} \setminus X'_j$
                    \State $X_{a^*} \leftarrow X'_j$
                    \State Go to Line \ref{while}
                \EndIf
            \EndFor
            
            \State Let $T_{\langle X,t \rangle}$ be the threshold-graph with $X=(X_1, \ldots, X_{n'})$ and $t=\frac{2}{3}(\mms_1, \ldots, \mms_{n'})$ for agents in $[n']$ \label{line:threshold}
            \State Let $M=\{(k+1,X_{k+1}), \ldots, (n',X_{n'})\}$ be a matching of size at least $1$ such that $N(\{X_{k+1}, \ldots, X_n\})=\{k+1, \ldots, n\}$ and $X_j$ is matched to $j$ for all $j \in [n]\setminus[k]$;\label{line:matching}
            \State $\N' \leftarrow [k]$;
            \State $\M \leftarrow \M \setminus \bigcup_{\ell \in [n] \setminus [k]} X_\ell$;
        \EndWhile    
        \Return $(X_1, X_2, \ldots, X_n)$;
    \end{algorithmic}
\end{algorithm}

\iffalse
\begin{lemma}
    During the execution of the while-loop in Algorithm~\ref{alg:approx-mms+efx}, write $S$ to denote the set of agents who has been assigned a non-empty set of bundles. Then, for any agent $i, j \in S$, we have $v_i(X_i) \geq \frac{2}{3}\mms_i$ and $i$ does not strongly envy any agent $j$.
\end{lemma}
\fi

\begin{restatable}{lemma}{matchingefx}\label{lem:matching_efx}
    For a given partition $X \in \Pi_n$ of $\M$ and a threshold vector $t=(t_1, \ldots, t_n)$, assume for all $j \in [n]$, there exists an agent $i$ such that $v_i(X_j ) \geq t_i$. Then $T_{\langle X,t \rangle}$ has a non-empty matching $M=\{(i_1, X_{j_1}), \ldots, (i_k, X_{j_k})\}$ such that $N(\{X_{j_1}, \ldots, X_{j_k}\})=\{i_1, \ldots, i_k\}$. Moreover, $M$ can be computed in polynomial time. 
\end{restatable}

\begin{proof}
    First we compute a maximum matching $M^*$ in $T_{\langle X,t \rangle}$ (which can be done in polynomial time).
	If $M^*$ is a perfect matching between $[n]$ and $(X_1, \ldots, X_{n})$, then clearly the lemma holds. Otherwise, there must exists a Hall's violator set $S \subset \{X_1, \ldots, X_n\}$ with $|N(S)| < |S|$. A minimal such set $S$ can be computed in polynomial time \cite{GAN2019104}. Note that for all $X_j \in S$, there exists an agent $i$ such that $v_i(X_j) \geq t_i$ and hence, $|S| \geq 2$. By minimality of $S$, the Hall's condition holds for any proper subset of $S$. Let $T \subset S$ and $|T| = |S|-1$. We have $|S|-1 = |T| \leq |N(T)| \leq |N(S)| < |S|$. Hence $|N(T)| = N(S) = |S|-1$. Since the Halls' condition holds for $T$, there exists a matching $M$ covering $N(T)$. Since $N(T) \subseteq N(S)$ and $|N(T)|=|N(S)|$, we have that $M$ is covering $N(S)$. Therefore, there exists no edge between the agents outside $M$ and the bundles inside $M$. 
\end{proof}

\begin{restatable}{theorem}{thmOne}\label{thm:efx+MMS-1}
    For any fair division instance with additive valuations, Algorithm \ref{alg:approx-mms+efx} returns a (partial) allocation that is both $\efx$ and $2/3$-$\mms$.
\end{restatable}
\begin{proof}
    We will begin by proving the correctness of Algorithm~\ref{alg:approx-mms+efx}, and then prove that it always terminates. 
    
    Consider any arbitrary iteration of the while-loop during the execution of Algorithm~\ref{alg:approx-mms+efx}. Let us assume there are $n'$ remaining agents at the start of this iteration. Without loss of generality, we can rename these remaining agents as $1, 2, \dots, n'$. This means that every agent $i \in [n]\setminus [n']$, has been assigned some bundle, say $X_i$.
    We begin by proving that $v_i(X_i) \geq \frac{2}{3}\mms_i$ and that agent $i$ does not strongly envy any agent $j \in [n]\setminus [n']$.
    
    We establish the above claim by induction. Since, initially no agent is assigned any bundle, the claim holds. Now, as the induction hypothesis, we assume that agents in $[n] \setminus [n']$ are already assigned a bundle and the (partial) allocation restricted to them is $2/3$-$\mms$ and $\efx$. Any change in the bundles as a result of the current while-loop can be examined by the following two cases: either the if-condition in Line \ref{if} is satisfied, and hence, only the bundle of agent $a^*$ changes in this iteration, or the if-condition in Line \ref{if} is not satisfied. 
    \begin{itemize}
        \item{\textit{Case 1: The if-condition in Line \ref{if} is satisfied.}} \label{case1}         
        First, note that, only the bundle of agent $a^*$ changes in this case. Let $X_{a^*}$ and $X'_{a^*}$ be the bundle of agent $a^*$ before and after this iteration of the while-loop respectively. By Line~\ref{line:a*}, we know that $v_{a^*}(X'_{a^*}) > v_{a^*}(X_{a^*}) \geq \frac{2}{3}\mms_{a^*}$, and hence, the allocation restricted to $[n]\setminus[n']$ is still $2/3$-$\mms$. Moreover, by the choice of $X'_{a^*}$ in Line~\ref{line:choice}, no agent in $[n]\setminus[n']$ strongly envies $X'_{a^*}$. Since $a^*$ did not strongly envy anyone while owning $X_{a^*}$, she still does not strongly envy anyone while owning $X'_{a^*}$. Hence, the allocation restricted to the set of agents in $[n] \setminus [n']$ is $\efx$ and $2/3$-$\mms$.\\
        
        \item{\textit{Case 2: The if-condition in Line \ref{if} is not satisfied.}} \label{case2}   
        Using Lemma~\ref{lem:matching_efx}, we know that the threshold-graph considered in Line~\ref{line:threshold} contains a matching $M$ of size at least one, such that, no unmatched agent has an edge to a matched bundle. 
      
        Now, without loss of generality, we rename the agents and bundles such that $M=\{(k+1, X_{k+1}), \ldots, (n', X_{n'}))\}$. Therefore, agents in the set $[n]\setminus[k]$ hold some non-empty bundle.
        Note that, by induction hypothesis and by the definition of the threshold-graph, we know that for all agents $i \in [n]\setminus[k]$, we have $v_i(X_i) \geq \frac{2}{3}\mms_i$. 
        
        Therefore, it remains to prove that the allocation restricted to agents in $[n]\setminus[k]$ is $\efx$ as well. We split these agents into the set $[n]\setminus[n']$ and $[n']\setminus[k]$. By induction hypothesis, we already know the allocation restricted to $[n]\setminus[n']$ is $\efx$. Next, since the if-condition in Line \ref{if} is not satisfied, no agent in $[n]\setminus[n']$ strongly envies any agent in $[n']\setminus[k]$. For all $i \in [n']\setminus[k]$, we have $v_i(X_i) \geq \frac{2}{3}\mms_i$ and $v_i(X_j) < \frac{2}{3}\mms_i$ for all $j \in [n]\setminus[n']$. Hence, no agent in $[n']\setminus[k]$ envies any agent in $[n]\setminus[n']$. Also, for all $i,j \in [n']\setminus[k]$ and all $X'_j \subsetneq X_j$, we have $v_i(X'_j) < \frac{2}{3}\mms_i$ (see Line \ref{line}). Since $v_i(X_i) \geq \frac{2}{3}\mms_i$, $i$ does not strongly envy $j$.
    \end{itemize}
     Finally, we will now prove that Algorithm~\ref{alg:approx-mms+efx} terminates and allocates a non-empty bundle to all agents. Let us write $A$ to denote the set of agents who are allocated a non-empty bundle at any point during the execution of Algorithm \ref{alg:approx-mms+efx}. We will prove that after each iteration of the while-loop, the vector $(\sum_{i \in A} v_i(X_i)$, $|A|)$ increases lexicographically, and hence, the algorithm must terminate. In Case 1, the utility of $a^*$ increases while the utility of all other agents in $A$ does not change and also $|A|$ does not change. On the other hand, in Case 2, since the matching $M$ found in Line~\ref{line:matching} is of size at least one, at least one more agent is added to the set $A$ and thus $|A|$ increases. Since, all agents who were previously in $A$, remain in $A$ and their utilities do not change, the claim follows. 
\end{proof}

Note that, the vector $(\sum_{i \in A} v_i(X_i)$, $|A|)$ can take pseudo-polynomially many values, and the only steps in Algorithm \ref{alg:approx-mms+efx} that cannot be executed in polynomial time are related to computing the exact $\mms$ values of agents and the construction of the bundles $X=(X_1, \ldots, X_{n'})$ such that $v_i(X_j) \geq (2/3)\mms_i$ in Line~\ref{line:mms_partition} of the while-loop. However, by Proposition \ref{prop:ptas} and Lemma \ref{lem:reduction-poly}, if we replace the $\mms$ bound $2/3$ with $2/3-\varepsilon$ for any constant $\epsilon>0$, these steps can be executed in polynomial time. Therefore, we obtain the following result.

\begin{theorem}\label{thm:pseudo-poly}
    For fair division instances with additive valuations and any constant $\varepsilon>0$, a (partial) allocation that is both $\efx$ and $(2/3-\varepsilon)$-$\mms$ can be computed in pseudo-polynomial time.
\end{theorem}

The only reason why the algorithm runs pseudo-polynomial time and not polynomial time, is that $\sum_{i \in A} v_i(X_i)$ in $(\sum_{i \in A} v_i(X_i)$, $|A|)$ can take pseudo-polynomially many values. By relaxing the notion of exact $\efx$ to $(1-\delta)$-$\efx$ for any constant $\delta$, we make sure that $v_i(X_i)$ can improve $\log_{1/(1-\delta)}(v_i(\M))$ many times which bounds the total number of rounds polynomially. 

\begin{theorem}\label{thm:poly}
    For fair division instances with additive valuations and any constant $\delta>0$ and $\varepsilon>0$, a (partial) allocation that is both $(1-\delta)$-$\efx$ and $(2/3-\varepsilon)$-$\mms$ can be computed in polynomial time.
\end{theorem}

\subsection{\boldmath Ensuring $2/3$-$\mms$ and $\efx$ with Charity}
In this section, we show that we can bound the number and the value of items that go unallocated in Algorithm~\ref{alg:approx-mms+efx}. We do so by using the algorithm $\mathtt{EFXwithCharity}$ developed by  \citet{chaudhury2021little} which takes a partial allocation $Y$ as input and outputs a (partial) $\efx$ allocation $X$ with the properties mentioned in Theorem \ref{thm:charity}. %which has fewer number of unallocated items (charity) and having less value as well. Furthermore, no agent can become worse off throughout the execution of this algorithm. Thus, if we run $\mathtt{EFXwithCharity}$ on the output of Algorithm \ref{alg:approx-mms+efx} (which is $\efx$ and $2/3$-$\mms$), we end up with a (partial) $\efx$ allocation which is still $2/3$-$\mms$ but also has all the properties that $\mathtt{EFXwithCharity}$ guarantees.

\begin{theorem}\cite{chaudhury2021little} \label{thm:charity}
    Given a (partial) $\efx$ allocation $Y$, there exists a (partial) $\efx$ allocation $X=(X_1, \ldots, X_n)$, such that for all $i \in [n]$
    \begin{enumerate}
        \item $X$ is $\frac{1}{2-|P(X)|/n}$-$\mms$, and
        \item $v_i(X_i) \geq v_i(Y_i)$, and 
        \item $v_i(X_i) \geq v_i(P(X))$, and
        \item $|P(X)|<s$,
    \end{enumerate}
    where $s$ is the number of sources in the envy-graph of $X$. 
\end{theorem}

Therefore, if we run $\mathtt{EFXwithCharity}$ on the output of Algorithm \ref{alg:approx-mms+efx} (which is $\efx$ and $2/3$-$\mms$), we end up with a (partial) $\efx$ allocation which is still $2/3$-$\mms$ but also has all the properties that $\mathtt{EFXwithCharity}$ guarantees.

\begin{restatable}{theorem}{thmTwo}\label{thm:efx+MMS-2}
    For any fair division instance with additive valuations, there exists a (partial) $\efx$ allocation $X=(X_1, \ldots, X_n)$ such that 
    \begin{enumerate}
        \item $X$ is $\max(2/3,\frac{1}{2-p/n})$-$\mms$, and
        \item for all $i \in [n]$, $v_i(X_i) \geq v_i(P(X))$, and
        \item $|P(X)|<s$,
    \end{enumerate}
    where $s$ is the number of sources in the envy-graph of $X$.
\end{restatable}

\begin{proof}
    Using Theorem \ref{thm:efx+MMS-1}, we know that there exist a (partial) $\efx$ allocation $Y$ which $2/3$-$\mms$. Then, we can use Theorem \ref{thm:charity} to obtain a (partial) allocation $X$ with all the stated properties.
\end{proof}

\section{\boldmath $\frac{2}{3}$-$\mms$ Together with $\ef1$}\label{sec:main-2}

In this section, we show that we can compute a complete allocation that is both $2/3$-$\mms$ and $\ef1$. Starting from the output of Algorithm \ref{alg:approx-mms+efx}, we run the well-known \emph{envy-cycle elimination procedure} \cite{lipton2004approximately} on the remaining items to obtain an $\ef1$ allocation which is $2/3$-$\mms$ as well; see Algorithm~\ref{alg:approx-mms+ef1}. We note that our result improves upon the previously best known approximation factor by \citet{amanatidis2020multiple} where they efficiently find allocations that are $4/7$-$\mms$ and $\ef1$.

The procedure of envy-cycle elimination was first introduced by \cite{lipton2004approximately} that computes an $\ef1$ allocation among agents having  monotone valuation; see Algorithm~\ref{alg:envy-cycle} for pseudocode. The idea is to start from an empty allocation and allocate the items one by one such that the partial allocation remains $\ef1$ in each round. In order to do so, one needs to look at the envy-graph of the allocation at each step of the algorithm. If it contains a cycle, by shifting the bundles along that cycle, the utility of all agents on that cycle improves, the allocation remains $\ef1$ and also the number of the edges in the envy-graph decreases. After removing all the cycles, the envy-graph must contain at least one source i.e., an agent whom no one envies. By allocating a remaining item to a source, the allocation remains $\ef1$. While originally, the algorithm starts with an empty allocation, one can also give a partial allocation as an input to the algorithm and perform the envy-cycle elimination procedure on the input allocation with remaining items. If the input allocation is $\ef1$, then the output allocation will be $\ef1$ as well. See Algorithm \ref{alg:envy-cycle} for the pseudo-code of our algorithm.

\begin{algorithm}[t]
    \caption{$\mathtt{envyCycleElimination}(\I, X)$}\label{alg:envy-cycle}
    \textbf{Input:} A fair division instance $\I = (\N, \M, \V)$ and a partial allocation $X=(X_1, \ldots, X_n, P)$\\
    \textbf{Output:} A complete allocation $X=(X_1, X_2, \ldots, X_n)$
    \begin{algorithmic}[1]
        \While{$P \neq \emptyset$}
            \While{there exists a cycle $i_1 \rightarrow i_2 \rightarrow \ldots \rightarrow i_k \rightarrow i_1$ in $G_X$}
                \State $A \leftarrow X_{i_1}$
                \For{$j \leftarrow 1$ to $k-1$}
                    \State $X_{i_j} \leftarrow X_{i_{j+1}}$
                \EndFor
                \State $X_{i_k} \leftarrow A$ 
            \EndWhile
            \State Let $s$ be a source in $G_X$
            \State Let $g$ be a good in $P$
            \State $X_s \leftarrow X_s \cup \{g\}$
            \State $P \leftarrow P \setminus \{g\}$
        \EndWhile        
        \Return $(X_1, X_2, \ldots, X_n)$;
    \end{algorithmic}
\end{algorithm}

The following lemma follows from the work of \cite{lipton2004approximately}.

\begin{restatable}{lemma}{envycycle}\label{lem:envy-cycle}
   Given an instance $\I$, if $X$ is a partial $\ef1$ allocation, then $\mathtt{envyCycleElimination}(\I, X)$ returns a complete $\ef1$ allocation $Y$ in polynomial time such that $v_i(Y_i) \geq v_i(X_i)$ for all $i \in [n]$.
\end{restatable}

\begin{proof}
      \cite{lipton2004approximately} showed that $\mathtt{envyCycleElimination}(\I, X)$ returns a complete $\ef1$ allocation in polynomial time.     Fix an agent $i$. In order to prove that $v_i(Y_i) \geq v_i(X_i)$, it suffices to prove that the value of agent $i$ never decreases throughout the algorithm. Initially, agent $i$ owns $X_i$. The bundle of agent $i$ only alters if we eliminate a cycle including agent $i$ which in that case agent $i$ receives a bundle which she envied before. Hence her utility increases. Another case is when agent $i$ is the source to whom we allocate a new good. Also in this case the utility of $i$ cannot decrease. Hence, in the end $v_i(Y_i) \geq v_i(X_i)$.
\end{proof}

\begin{algorithm}[t]
    \caption{$\mathtt{approxMMSandEF1}(\I)$}\label{alg:approx-mms+ef1}
    \textbf{Input:} A fair division instance $\I = (\N, \M, \V)$ \\
    \textbf{Output:} A complete allocation $X$
    \begin{algorithmic}[1]
        \State $X \leftarrow \mathtt{approxMMSandEFX}(\I)$
        \State $X \leftarrow \mathtt{envyCycleElimination}(\I, X)$
        \Return $X$; 
    \end{algorithmic}
\end{algorithm}

We now prove our next result that deals with the compatibility of $\ef1$ allocations with $\mms$ guarantees. 

\begin{theorem} \label{thm:ef1+MMS}
    For fair division instances with additive valuations, Algorithm~\ref{alg:approx-mms+ef1} returns a complete allocation which is $\ef1$ and $2/3$-$\mms$.
\end{theorem}
\begin{proof}
     For a given fair division instance, Algorithm~\ref{alg:approx-mms+ef1} begins by running Algorithm~\ref{alg:approx-mms+efx} as a subroutine. By Theorem~\ref{thm:efx+MMS-1}, we know that $\mathtt{approxMMSandEFX}(\I)$ returns a partial allocation $X$ which is $2/3$-$\mms$ and $\efx$ and thus $\ef1$. Then, it runs envy-cycle elimination with the remaining items. And, by Lemma~\ref{lem:envy-cycle}, we know that $\mathtt{envyCycleElimination}(\I, X)$ returns a complete allocation $Y$ which is $\ef1$. Moreover, Lemma~\ref{lem:envy-cycle} shows that $v_i(Y_i) \geq v_i(X_i)$ for all agents $i$. Since $X$ is a $2/3$-$\mms$ allocation, $Y$ continues to be a $2/3$-$\mms$ allocation as well. This completes our proof.
\end{proof}

Note that, the envy-cycle elimination procedure runs in polynomial time. For any constant $\varepsilon>0$ and $\delta>0$, by Theorem \ref{thm:pseudo-poly}, we can compute a complete a $(2/3-\varepsilon)$-$\mms$ and $\ef1$ allocation in pseudo-polynomial and by Theorem \ref{thm:poly}, we can compute a $(2/3-\varepsilon)$-$\mms$ and $(1-\delta)$-$\ef1$ allocation in polynomial time.

\begin{theorem}\label{thm:polyEF1}
    For fair division instances with additive valuations and any constants $\varepsilon>0$ and $\delta>0$, a complete allocation that is both $\ef1$ and $(2/3-\varepsilon)$-$\mms$ can be computed in pseudo-polynomial time and a complete allocation that is both $(1-\delta)$-$\ef1$ and $(2/3-\varepsilon)$-$\mms$ can be computed in polynomial time.
\end{theorem}

\iffalse
\citeauthor{Caragiannis2023} shows that every $\mms$ allocation is $\eefx$ in any fair division instance with additive valuations. It is not difficult to see that in fact for any $\alpha \in [0,1]$, every $\alpha$-$\mms$ allocation is $\alpha$-$\eefx$ (see Lemma \ref{lem:mms-eefx}).  \HA{Do we need to prove it or do you have it in your eefx paper?}
\begin{corollary}[of Theorem \ref{thm:ef1+MMS}] \label{thm:eefx+ef1}
    For any fair division instance with additive valuations, there exists a complete allocation that is both $\ef1$ and $2/3$-$\eefx$.
\end{corollary}
\fi

\section{Conclusion} \label{sec:conc}
In this work, we embark upon pushing our understanding of achieving guarantees for $\mms$ with $\efx/\ef1$ notions of fairness. We improve the approximation guarantees for the above by developing pseudo-polynomial time algorithms to compute, for any constant $\varepsilon>0$, (i) a partial allocation that is both $(2/3-\varepsilon)$-$\mms$ and $\efx$, and  (ii) a complete allocation that is both $(2/3-\varepsilon)$-$\mms$ and $\ef1$. 
 
While enhancing the above fairness guarantees, we develop a new technique, via Algorithm \ref{alg:approx-mms+efx}, for finding desired partial $\efx$ allocations, in particular, where we have a provable good bound on the amount of value each agent receives. An important line for future work is to further improve the simultaneous guarantees for achieving fairness notions of $\mms$ with $\efx$/$\ef1$.

\bibliographystyle{plainnat}
\bibliography{references}

\end{document}